\documentclass[conference]{IEEEtran}
\IEEEoverridecommandlockouts

\usepackage{cite}
\usepackage{amsmath,amssymb,amsfonts,amsthm}
\usepackage{algorithmic} 
\usepackage{graphicx}
\usepackage{textcomp}
\usepackage{tikz}
\usepackage{graphicx}
\usepackage{pgfplots}
\usepackage{mathrsfs}  
\usepackage[shortlabels]{enumitem}
\usepackage[normalem]{ulem}
\usepackage{bbm}
\usetikzlibrary{spy,backgrounds}

\newtheorem{theorem}{Theorem}
\newtheorem{lemma}{Lemma}

\newtheorem{definition}{Definition}

\theoremstyle{remark}
\newtheorem{remark}{Remark}

\usepackage{tikz}
\usepackage{caption}
\usepackage{float}
\usetikzlibrary{decorations.pathreplacing}

\newenvironment{psmallmatrix}
  {\left(\begin{smallmatrix}}
  {\end{smallmatrix}\right)}

\newcommand\numeq[1]%
  {\stackrel{\scriptscriptstyle(\mkern-1.5mu#1\mkern-1.5mu)}{=}}
  \newcommand\numl[1]%
  {\stackrel{\scriptscriptstyle(\mkern-1.5mu#1\mkern-1.5mu)}{<}}
\newcommand\numleq[1]%
  {\stackrel{\scriptscriptstyle(\mkern-1.5mu#1\mkern-1.5mu)}{\leq}}
\newcommand\numgeq[1]%
  {\stackrel{\scriptscriptstyle(\mkern-1.5mu#1\mkern-1.5mu)}{\geq}}
\usepackage{mathtools}

\DeclarePairedDelimiterX\braket[2]{\langle}{\rangle}{#1 \delimsize\vert #2}

\usepackage{breqn}

\usepackage{xcolor}
\def\BibTeX{{\rm B\kern-.05em{\sc i\kern-.025em b}\kern-.08em
    T\kern-.1667em\lower.7ex\hbox{E}\kern-.125emX}}
    
    \usepackage{cuted}
\begin{document}

\title{Capacity Achieving Codes for an Erasure Queue-channel 
}

\author{\IEEEauthorblockN{Jaswanthi Mandalapu\IEEEauthorrefmark{1}, Krishna Jagannathan\IEEEauthorrefmark{1}\IEEEauthorrefmark{2}, Avhishek Chatterjee\IEEEauthorrefmark{1}, Andrew Thangaraj\IEEEauthorrefmark{1}}
\IEEEauthorblockA{\IEEEauthorrefmark{1} Department of Electrical Engineering, IIT Madras  \\
\IEEEauthorrefmark{2} Centre for Quantum Information, Communication and Computing (CQuICC),  IIT Madras}
}

\maketitle
\begin{abstract}

We consider a queue-channel model that captures the waiting time-dependent degradation of information bits as they wait to be transmitted. Such a scenario arises naturally in quantum communications, where quantum bits tend to decohere rapidly. Trailing the capacity results obtained recently for certain queue-channels, this paper aims to construct practical channel codes for the erasure queue-channel (EQC)---a channel characterized by highly correlated erasures, governed by the underlying queuing dynamics. Our main contributions in this paper are twofold: (i) We propose a generic `wrapper' based on interleaving across renewal blocks of the queue to convert any capacity-achieving block code for a memoryless erasure channel to a capacity-achieving code for the EQC. Next, due to the complexity involved in implementing interleaved systems, (ii) we study the performance of LDPC and Polar codes \emph{without} any interleaving. We show that standard Arıkan's Polar transform polarizes the EQC for certain restricted class of erasure probability functions. We also highlight some possible approaches and the corresponding challenges involved in proving polarization of a general EQC.
 
\end{abstract}
\section{Introduction}
A \textit{`queue-channel'} is a model that captures the waiting time degradation of information bits in a queue as they wait to be processed. The motivation for such a scenario arises naturally in quantum communications, where the inevitable buffering of qubits at intermediate nodes causes them to suffer waiting time-dependent decoherence \cite{nielsen2002quantum}. In other words, the longer a qubit waits in the buffer, the more likely it decoheres, leading to the loss of information. Similar issues also arise in delay-sensitive applications such as multimedia streaming and Ultra Reliable Low Latency Communications (URLLC), where information bits become less useful after a certain time \cite{avhishek}.

This paper considers a queue-channel setting where information bits are encoded into codewords and transmitted sequentially over a single-server queue. We assume that as the bits await processing in the queue, they become less valuable in time, leading to erasures; see Fig.~\ref{queue-channel}. In particular, the erasure probability of an information bit $i$ is modeled by an explicit function of its sojourn time in the system---i.e., if $W_i$ is the sojourn time of  the $i$th bit, its erasure probability is given by $p(W_i)$, where $p(\cdot)$ is typically a non-decreasing function. Despite the simplicity of the model, a queue-channel is challenging to analyze from an information theoretic perspective, for the following reasons: (i) the channel exhibits memory because of the underlying Markov nature of waiting times, and (ii) if the queue is assumed to be initially empty, then the channel is non-stationary. Recent literature \cite{jagannathan2019qubits, mandayam2020classical, siddhu2021unital} has characterized the single-letter capacity expression of a queue-channel for the cases of erasures and other noise models. Following these capacity results, in this work, we aim to construct practical channel codes for an \emph{erasure queue-channel} (EQC).
\tikzset{meter/.append style={draw, inner sep=10, rectangle, font=\vphantom{A}, minimum width=30, line width=.8,
 path picture={\draw[black] ([shift={(.1,.3)}]path picture bounding box.south west) to[bend left=50] ([shift={(-.1,.3)}]path picture bounding box.south east);\draw[black,-latex] ([shift={(0,.1)}]path picture bounding box.south) -- ([shift={(.3,-.1)}]path picture bounding box.north);}}}
\begin{figure}
\hspace{-0.3in}
\begin{tikzpicture}
\draw[-,dashed] (0.3,1.8) -- (9.2,1.8) -- (9.2,-0.8) -- (0.3,-0.8) -- (0.3,1.8);
    \draw[->] (0.5,0.5) -- (1,0.5) node[above,midway,text width=20mm,text centered]{$m$};
    \draw[->] (8.5,0.5) -- (9,0.5) node[above,midway,text width=20mm,text centered]{$\hat{m}$};
    \draw (1,-0.2) rectangle (3,1.2) node[midway,text width=20mm,text centered]{Encoder $\lambda$ bits/sec};
    \draw[->] (3,0.5) -- (3.7,0.5) node[above,midway,text width=20mm,text centered]{$X^N$};
    \draw[-] (3.7,0.1) -- (5,0.1) -- (5,0.9) -- (3.7,0.9);
    \node at (4.3,0.5) {$ ? ? 1 ? 0$};
   \draw (5.35,0.5) circle [radius = 0.35] node {$\mu$};
   \draw[->] (5.7,0.5) -- (6.4,0.5) node[above,midway,text width=20mm,text centered]{$Y^N$};
   \draw (6.4,-0.2) rectangle (8.5,1.2) node[midway,text width=20mm,text centered]{Decoder };
\end{tikzpicture}
    \caption{Schematic of an erasure queue-channel}
    \label{queue-channel}
\end{figure}
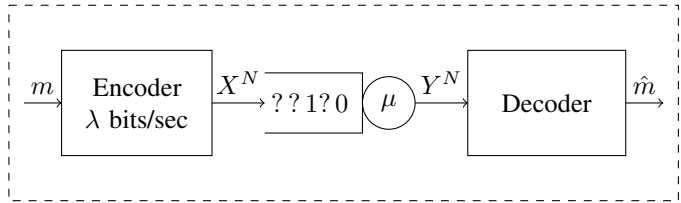

\subsection{Our Contributions}
 The main contributions in this paper are two-fold. Firstly, we provide a generic coding `wrapper' that converts any capacity-achieving code for an i.i.d. erasure channel to a code that achieves the capacity for an erasure queue-channel. The wrapper technique relies primarily on the fact that information bits that are `sufficiently far apart' in the queue tend to experience `nearly independent' channels. Specifically, due to the `renewal' nature of queuing systems, the information bits falling in different busy periods of the queue see independent channels. Exploiting this property, we propose a coding wrapper that consists of an interleaver and de-interleaver, applied over the i.i.d. erasure channel code. We derive a key concentration bound on the number of bits falling in a particular busy period of the queue, which then helps us in characterizing a lower bound on the interleaving length. Although the wrapper technique works in principle, there are practical disadvantages involved in implementing interleaved coding schemes --- they suffer from increased latency, as well as larger memory and computational requirements.

Hence, in the second part of this paper, we focus on analyzing the performance of a few stronger codes \textit{without} any interleaving on erasure queue-channels. We first evaluate numerically the performance of LDPC and Polar codes, used `as is' on an EQC. Numerical results indicate that both LDPC and Polar codes achieve rates very close to capacity with low block error probability. Encouraged by the good empirical performance of Polar codes, we next aim to derive theoretical guarantees of Arıkan's polar codes over an EQC. Prior works \cite{sasoglu,sasoglu2011polar} show that standard polar coding transform can be applied directly on a broad class of `fast-mixing' channels with memory to achieve the capacity. In particular, Şaşoğlu and Tal \cite{sasoglu} proved that a class of stationary memory channels that are `promptly $\psi-$mixing' do polarize under the standard Arıkan polar construction \cite{arikan}. It is also shown in \cite{shuval2018fast} that all finite-state Markov channels satisfy $\psi-$mixing. However, in the queue-channel model we consider, the erasure events are governed by the sojourn time dynamics, which has Markovian evolution in an uncountable state space. It appears challenging to  establish directly the promptly $\psi-$mixing property for the EQC. Owing to the above technical challenges,  we  prove the polarization for EQC under the technical restriction that the erasure probability sequence $\{p(W_i), i\ge 1)\}$ has Markovian evolution in a finite state space.


Although the motivation for queue-channel arises primarily from quantum communications, our work  focuses solely on constructing \emph{classical} channel codes for an EQC. This is mainly because 
the classical capacity of a quantum EQC is exactly the capacity of its induced classical channel \cite{mandayam2020classical}; in particular, encoding classical information bits into untangled orthogonal quantum states can achieve the classical capacity for a quantum EQC. 

The rest of this paper is organized as follows: In Sec.~\ref{preliminaries}, we briefly review the model of an erasure queue-channel and its capacity. In Sec.~\ref{generic_wrapper}, we describe a generic coding wrapper to convert any capacity-achieving code of an i.i.d. erasure channel to a code that achieves capacity for an erasure queue-channel. Next, Sec.~\ref{coding} provides the numerical analysis of LDPC and Polar codes over an erasure queue-channel. Following the numerical results, Sec.~\ref{polar_theoretical} provides the theoretical guarantees of Arıkan's polar codes for a class of technically restricted erasure queue-channels, and discuss the open challenges. Finally, Sec.~\ref{conclusion} concludes the paper. Detailed proofs of the results stated in this paper are provided in Sections~\ref{polarization_proof} and \ref{appendix}.

\section{The Erasure Queue-channel and its Capacity}\label{preliminaries}
In this section, we review the framework of an erasure queue-channel (EQC) introduced in \cite{jagannathan2019qubits} and \cite{mandayam2020classical}. A source generates an input message $m$ and encodes it into a coded bit sequence $X^N$. These coded bits are then transmitted sequentially to a single-server queue according to a continuous-time stationary point process of arrival rate $\lambda$. The server serves the information bits with independent and identically distributed (i.i.d.) service times with mean $1/\mu$ in a First-Come-First-Served (FCFS) service discipline. The arrival process is assumed to be independent of the service time process, and for the stability in the queue, we assume $\lambda < \mu$.

In an erasure queue-channel, the probability of erasure of a particular bit is modeled as a function of its waiting time. Specifically, let $W_i$ denote the total sojourn time\footnote{In queuing literature, sojourn time indicates the total time spent by an information bit in the queue, i.e., the time including the waiting time and the service time.} of the $i$th information bit in the queue. Then the erasure probability of this bit is modeled as an explicit function of its sojourn time, denoted by $p(W_i)$. Note that the function $p:[0,\infty] \to [0,1]$ is typically a non-decreasing function of the sojourn time $W$. At the decoder, a (partially erased) coded sequence $Y^N$ is received, from which the output message $\hat{m}$ is estimated. Precisely, an $N-$length transmission over an erasure queue-channel is defined as follows: Inputs are $\{X_i, 1 \leq i \leq n\}$ chosen from the input alphabet set $\mathcal{X}$, channel distribution $\Pi_i P(Y_i|X_i,W_i)$, and outputs are $\{Y_i, 1 \leq i \leq n\}$ belongs to the output alphabet set $\mathcal{X} \cup \{e\}$, where $e$ represents an erasure. Fig.~\ref{queue-channel} depicts the detailed schematic of the system under study.

As observed in \cite{jagannathan2019qubits,mandayam2020classical}, an EQC is neither a stationary nor a memoryless channel. Indeed, the erasure probability of any information bit $i$ depends on its sojourn time $W_i,$ which is governed by Lindley's recursion \cite[Page~239]{gallager2013stochastic}. That is, if $A_i$ denotes the inter-arrival time between the information bits $i-1$ and $i$, and $S_i$ denotes the service time of the $i^{th}$ information bit in the queue, then 
    \begin{align}\label{lindleys_recursion}
        W_{i+1} = \max(W_i - A_{i+1}, 0) + S_{i+1}.
    \end{align} 
  Following this, the definition of capacity and its characterization have been provided for an EQC in \cite{jagannathan2019qubits} and \cite{mandayam2020classical}. We recall the derived capacity result here for brevity.
  
    \begin{theorem} \cite{jagannathan2019qubits} \label{capacity}
        The capacity of an erasure queue-channel is given by $\lambda \mathbb{E}_{\pi}[1 - p(W)] $bits/sec, where $\pi$ is the stationary distribution of the sojourn times of the bits in the queue.
    \end{theorem}
   
    Often in many quantum systems, a practically well-motivated model for $p(W)$ is an exponential function, i.e., $p(W)={1 - e^{-\kappa W}}$. The constant $\kappa$ here can be referred to as the decoherence parameter, which usually depends on the physical or implementation parameters such as temperature. In such a case, it can be easily seen that the above capacity expression in bits per channel use is reduced to the Laplace transform of the sojourn time evaluated at $\kappa$; Fig.~\ref{capacity_num} depicts this capacity for $\mu = 1$.
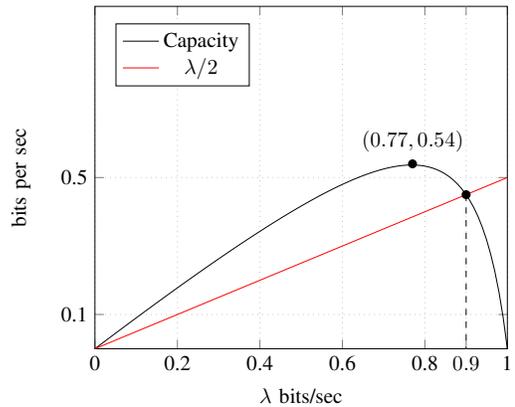
\begin{figure}[!t]
\centering
    \begin{tikzpicture}[scale = 0.8]
        \begin{axis}[
        xlabel = {\(\lambda\) bits/sec},
        ylabel = {bits per sec},
       xtick={0,0.2,0.4,0.6,0.8,1},
       xticklabels={0,0.2,0.4,0.6,0.8,1},
       ytick={0.1,0.5},
       yticklabels={0.1,0.5},
       extra x ticks={0.9},
extra x tick label={$0.9$},
        xmin=0, xmax=1,
    ymin=0, ymax=1, 
    legend entries = {Capacity,$\lambda/2$},
        grid = major,grid style = dotted,legend style={at={(0.05,0.95)},anchor=north west}
    ]
      \addplot [
        domain=0:1, 
        samples=100, 
        color=black
        ]
     {(x*(1-x)/(1.1-x))};\label{no_coding}
     \addplot[domain=0:1, 
        samples=100, 
        color=red
        ]
     {x/2};\label{lambda/2}
    \node[fill,circle,inner sep=1.5pt] at
(90,45){};
    \node[fill,circle,inner sep=1.5pt,label={above:$(0.77,0.54)$}] at
(77,54){};
    \draw[-,dashed] (90,45) -- (90,0);
    \end{axis}
    \end{tikzpicture}
    \caption{Capacity of an EQC}%
    \label{capacity_num}
\end{figure}


\section{A Generic Coding Wrapper for Erasure-queue channel}\label{generic_wrapper}
In this section, we present a generic coding `wrapper' that can convert any capacity-achieving code for i.i.d. erasure channel with erasure probability $\mathbb{E}_{\pi}[p(W)]$, such as Polar, Reed-Muller, and SC-LDPC codes into a code that can achieve capacity for an erasure queue channel.

\subsection{Encoder and Decoder for an erasure Queue-channel}
Consider the setting where a message block of length $\bar{K}$ has to be transmitted over an erasure 
queue-channel. Our goal is to design an erasure queue-channel code\footnote{An analysis of simple two-one repetition code for an EQC has been provided in Appendix Sec.\ref{two-onerep}.} whose rate and probability of error tend to $C(\lambda)$ and $0$, respectively, as ${\bar{K} \to \infty}$.

\textit{Encoding Scheme:}
Let $\mathcal{C}$ be a $(N,K)$ code for i.i.d. erasure channel with erasure probability $\mathbb{E}_{\pi}[p(W)]$. 

The queue-channel encoder divides the original message block of length $\bar{K}$ into $B$ message blocks, each of length $K$:  $\{ m^{(i)}: 1 \le i \le B\}$. Note that $\bar{K} = KB,$ and without loss of generality, we assume $\bar{K}$ to be divisible by $B$. Next, for each $i$, $1\le i \le B$, it maps each message block $m^{(i)}$ to a codeword  $c^{(i)}$ of length $N$ using the $(K,N)$ code $\mathcal{C}$. Then, the encoder transmits the following $NB$ length codeword
\begin{equation}
c_1^{(1)},c_1^{(2)},\ldots, c_1^{(B)}, c_2^{(1)},c_2^{(2)},\ldots, c_2^{(B)}, \ldots, c_N^{(1)},c_N^{(2)},\ldots, c_N^{(B)},\nonumber
\end{equation}
where $c^{(i)}_j$ is the $j$th symbol of the codeword $c^{(i)}$.

\textit{Decoding Scheme:} The decoder arranges the received $NB$ symbols $y_{1:NB}$ in a $N \times B$ matrix, say $\mathbf{Y}$. Then, it decodes the $i^{th}$ column of matrix $\mathbf{Y}$, which is $\{y_{1i},y_{2i},\ldots,y_{Ni}\}$, using the optimal decoder for the $(K,N)$ code $\mathcal{C}$ and obtains $\{\hat{m}^{(i)}: 1\le i \le B\}$ as the estimates of $\{  m^{(i)}: 1 \le i \le B\}$. The decoder finally outputs 
\begin{equation}
\hat{m}^{(1)}_1, \hat{m}^{(1)}_2, \ldots, \hat{m}^{(1)}_K, \hat{m}^{(2)}_1, \hat{m}^{(2)}_2, \ldots, \hat{m}^{(B)}_1, \hat{m}^{(B)}_2, \ldots, \hat{m}^{(B)}_K.
\end{equation}

Consider that the code $\mathcal{C}$ achieves capacity for an i.i.d. erasure channel with erasure probability $\mathbb{E}_{\pi}[p(W)]$. Then, the above encoding and decoding scheme has the following performance guarantee.

\begin{theorem}\label{generic_code}
    The above encoding and decoding scheme achieves a rate arbitrarily close to $C(\lambda) = \lambda \mathbb{E}_{\pi}[1-p(W)]$ for any arrival rate $\lambda$ to an $M/M/1$\footnote{In an $M/M/1$ queue, the arrivals are determined by Poisson process, and the job service times are exponential} queue; refer to Appendix Sec.~\ref{gen_queue_bound} queue. \\
    In particular,
    \begin{itemize}
        \item If $P_{\gamma}$, which is the decoding error probability of code $\mathcal{C}$ scales according to $O(e^{-N^{\beta}})$ for some $\beta > 0$, and
        \item for any $\alpha < k_1 e^{-N^{\beta}}$, where $k_1$ is a positive constant if $B > \frac{\ln{1/\alpha}}{2\ln{(\lambda+\mu)} - \ln{(4\lambda \mu)}}$, 
    \end{itemize}
    then the coding scheme proposed above achieves any rate $R < C(\lambda)/\lambda$ for sufficiently large $N$ with error probability $P_e$ scaling according to $O((N+B)e^{-N^{\beta}})$. 
\end{theorem}

The following concentration bound on the number of information bits that can be processed during a busy cycle of a $G/G/1$ queue will be an essential ingredient in proving Theorem~\ref{generic_code}.
\subsection{The queuing bound}
 \begin{definition}
      A busy period in a single-server queue is the time between the arrival of a job at an empty queue and the queue becoming empty again.
   \end{definition}

A \emph{renewal point} is a point at which a job arrives at an empty queue in a single-server queue. It signifies the point from which the queue process, in a sense, restarts and remains independent of the past till then. The time between two consecutive renewal points is called a renewal cycle. Furthermore, the jobs that arrive in two different renewal cycles of a queue experience independent waiting times. 

   By the definition of a busy period, two jobs that arrive in different renewal cycles are also in two different busy periods. Thus, the smaller the number of jobs arriving in a busy period weaker the dependence across jobs. Consequently, this would imply a weaker dependence between erasures experienced by the symbols passing through a queue-channel. The following lemma quantifies this intuition dependence (or independence) by obtaining a bound on the number of arrivals in a busy period. For simplicity in analysis, we provide our results here for an $M/M/1$ for the same concentration bound for any general queue in the system.

    \begin{lemma}\label{concentration}
      For any $n>0$, the total number of arrivals $J({\dot{T}_b})$ in a `typical' busy period $\dot{T}_b$ follows
      ${
          P\{J(\dot{T}_b) > l\} \leq \left(\frac{\lambda \mu}{(\frac{\lambda + \mu}{2})^2}\right)^{l}.
      }$
   \end{lemma}
 \begin{proof}
     Refer to Sec.\ref{concentrationlemmaproof}.
 \end{proof}

\subsection{Proof of Theorem~\ref{generic_code}}
Using Lemma~\ref{concentration}, we now provide the proof of Theorem~\ref{generic_code}.

Let $\Phi$ be the event that for a sequence of $N$ bits from a chosen column of the received information matrix $\mathbf{Y}$, there is no busy period with number of bits more than $B$. Note that if $\Phi$ happens then for any $i$, the $i$th and $(i+B)$th bit are in two different busy periods and hence, they experience independent waiting based erasures. Here, we first upper bound $P(\Phi^c)$.

Note that as the busy periods are independent and there can be at most $N$ busy periods, when we consider a sequence of $N$ jobs, we get the following.
\begin{align*}
P(\Phi^c) 
& \le P(\mbox{\# of bits in at least  one of } N \mbox{ busy periods }  \ge B) \\
&\numleq{a} N P(\mbox{\# of bits in a busy period } \ge B)  
\numleq{b} N \alpha,
\end{align*}
where $(a)$ is true from the Union bound, and $(b)$ follows from Lemma~\ref{concentration} and the choice of $B$. Note that $\alpha = \left(\frac{\lambda \mu}{(\frac{\lambda + \mu}{2})^2}\right)^{B}$ can be arbitrary small based on value of $B.$

Further, the achievable rate $R$ using the proposed coding scheme for any chosen $\lambda$ would be
\begin{align}\label{success_prob}
\begin{split}
    R &\numeq{a1}   \frac{1}{NB}[B(N(1 - \mathbb{E}_{\pi}[p(W)])) P(\Phi) + P(\Phi^c) R']\\
      &= 1 - \mathbb{E}_{\pi}[p(W)] - \mu(\alpha),
\end{split}
\end{align}
where $\mu(\alpha) \to 0$ as $\alpha \to 0$. Note that $R'$ in $(a1)$ is the rate that can be achieved by the proposed coding scheme when the event $\Phi^c$ occurs. In addition, if $P_{\gamma}$ is the decoding error probability of code $\mathcal{C}$, then the decoding error probability $P_e$ for the proposed coding scheme is upper-bounded as follows:
\begin{align*}
    P_e &\leq P(\Phi^c) + P(\Phi) B P_{\gamma}
        \numleq{a2} N\alpha + BP_{\gamma},
\end{align*}
where $(a2)$ is true from equation~\eqref{success_prob}. Finally, the above expression implies, that if $P_{\gamma}$ scales according to $O(e^{-N^{\beta}})$ and $\alpha < k_1 e^{-N^{\beta}}$ for some $\beta > 0$, then $P_e$ scales in the order $O((N+B)e^{-N^{\beta}})$ completing the proof.\qed

\begin{remark}
   Although the above coding strategy achieves the capacity of an erasure queue-channel, in general, the performance of interleaved systems may be poor because of the following reasons: (i) low latency, as the decoder has to wait for the other blocks in order to decode a message block, (ii) it requires higher memory to store the data, and (iii) extra computational blocks for interleaving and de-interleaving. Therefore, in our subsequent sections, we numerically analyze the performance of conventional Polar and LDPC codes \emph{without} any interleaving. Furthermore, we theoretically analyze if standard Arıkan's Polar codes \cite{arikan2009channel} can achieve rates close to capacity for an erasure queue-channel.
\end{remark}

\section{Numerical Analysis of LDPC and Polar Codes}\label{coding}
 \begin{figure}[!b]
    \centering
    \begin{tikzpicture}[scale = 0.8]
        \begin{semilogyaxis}[
		grid = minor, grid style = dotted,xlabel= $\log_2N$,ylabel=Block Error Probability,  legend entries = {$\lambda = 0.8$, $\lambda = 0.85$, $\lambda = 0.87$}, legend style ={at={(0.05,0.05)},anchor=south west} ]	
 \addplot[color=black,mark=*] coordinates {
		(6.58,5.3e-02)
		(7.6724,4.76e-02)
		(8.6724,4e-03)
		(9.9773,4e-04)
		(11.3663,1e-05)
		(12.9658,0)
	};
 \addplot[color=blue,mark=*] coordinates {
		(6.58,2.85e-01)
		(7.6724,2.7e-01)
		(8.6724,1.3e-01)
		(9.9773,4e-02)
		(11.3663,4e-03)
		(12.9658,0)
	};
  \addplot[color=red,mark=*] coordinates {
		(6.58,5.02e-01)
		(7.6724,4.51e-01)
		(8.6724,4.28e-01)
		(9.9773,3.64e-01)
		(11.3663,1.99e-01)
		(12.9658,1.9e-02)
	};
	\end{semilogyaxis}%
    \end{tikzpicture}
    \caption{Performance of rate half LDPC coding with an interleaver and SPA decoder on an EQC.}
    \label{ldpc_interleaving}
\end{figure}
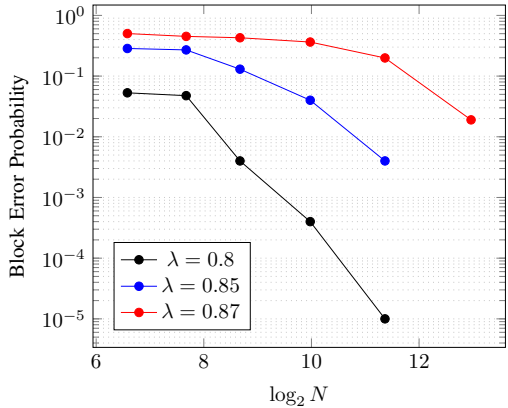

For numerical analysis, we consider the following system parameters:
\begin{enumerate}
    \item We consider an $\mathsf{M}/\mathsf{M}/1$ queue in the system, i.e., the inter arrival times are exponential and the service times are exponential with $\mu =1$.
    \item The erasure probability function is considered to be exponential, i.e., $p(W) = 1 - e^{-\kappa W}$ with $\kappa = 0.1$
\end{enumerate}
\textit{Numerical Evaluation of LDPC Codes:} Using the base matrices available from \cite{Mackay}, we now numerically evaluate the performance of LDPC codes on an EQC. For ease of illustration, we consider a rate half LDPC code. From Fig.~\ref{capacity_num}, we see that at $\lambda = 0.9$, $\mathrm{C}(\lambda) = 0.45$ in bits/sec and consequently, the capacity is $0.5$ in bits per channel use. Further, when arrival rates are $0.8, 0.85, $ and $0.87$, the capacities are $0.669, 0.6$ and $0.56$ bits per channel use, respectively. 
Now, following Sec.\ref{generic_wrapper}, Fig.~\ref{ldpc_interleaving} depicts the performance of rate half LDPC code with an interleaver on an EQC. We observe that as the block length increases, the wrapper technique can achieve rates close to capacity with low error probabilities. We \emph{indeed} observe that for larger block lengths, LDPC codes, even without interleaver, can achieve capacity with arbitrarily low error probability (Fig.~\ref{ldpc_performance}). 

\noindent
\textit{Numerical Evaluation of Polar Codes:}
Similar to Fig.\ref{ldpc_interleaving},  Fig.~\ref{polar_code_results_interleaving} depicts the performance of polar codes with interleaving, where the performance is evaluated at optimal arrival rate $\lambda*$ on an EQC. Following the results from \cite{jagannathan2019qubits}, when $\kappa = 0.1$, we see that $\lambda^* = 0.77$ and $\mathrm{C}(\lambda^*) = 0.54$ bits per sec, which consequently implies that the capacity is $0.7$ bits per channel use. We observe that the error probabilities reduce with increasing block lengths for the rates closer to capacity. On the other hand, Fig.~\ref{polar_code_results} illustrates that conventional Arıkan's polar transform with successive cancellation decoding also achieves rates close to capacity with low error probabilities.

Encouraged by the numerical results, it seems worthwhile to investigate analytically the performance of LDPC and Polar codes for the EQC. In the next section, we indicate a possible approach to proving polarization for the EQC and highlight certain technical challenges that remain to be resolved.

\begin{figure}[b!]
    \centering
    \begin{tikzpicture}[scale = 0.8]
        \begin{semilogyaxis}[
		grid = minor, grid style = dotted,xlabel= $\log_2N$,ylabel=Block Error Probability,  legend entries = {$\lambda = 0.8$, $\lambda = 0.85$, $\lambda = 0.87$}, legend style ={at={(0.05,0.05)},anchor=south west} ]	
 \addplot[color=black,mark=*] coordinates {
		(6.58,2.17e-01)
		(7.6724,1.76e-01)
		(8.6724,1.23e-01)
		(9.9773,3e-02)
		(11.3663,1e-03)
		(12.9658,0)
	};
 \addplot[color=blue,mark=*] coordinates {
		(6.58,3.77e-01)
		(7.6724,3.51e-01)
		(8.6724,3e-01)
		(9.9773,2.68e-01)
		(11.3663,1.7e-01)
		(12.9658,9e-02)
	};
  \addplot[color=red,mark=*] coordinates {
		(6.58,4.42e-01)
		(7.6724,4.24e-01)
		(8.6724,4.26e-01)
		(9.9773,4.13e-01)
		(11.3663,3.91e-01)
		(12.9658,3.12e-01)
	};
	\end{semilogyaxis}%
    \end{tikzpicture}
    \caption{Performance of conventional rate half LDPC coding with SPA decoder on an EQC.}
    \label{ldpc_performance}
\end{figure}
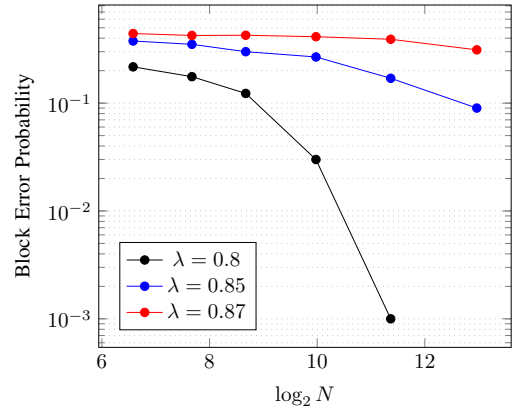

\section{Does the Erasure Queue-channel Polarize ?}\label{polar_theoretical}

We now prove the polarization of an EQC under the Arıkan transform for a certain restricted class of the erasure probability functions. In particular, if the erasure probability sequence $\{p(W_i), i \geq 1\}$ follows a finite-state Markov chain, we can invoke the existing results on $\psi$-mixing from \cite{sasoglu}, and assert the capacity achieving nature of polar codes for this restricted class of EQCs. Accordingly, Theorems~\ref{error_b} through \ref{fast_polarization} are proved for EQCs with the above restriction on the erasure probabilities in Section~\ref{polarization_proof}.

\begin{theorem}\label{error_b}
    For an erasure queue channel, and for a given arrival rate $\lambda$, standard Arıkan's polar construction achieves any rate $\mathrm{R}$ such that $\lambda \mathrm{R} < \mathrm{C}(\lambda)$, for sufficiently large enough block length $N$. Further, the block error probability scales according to ${P_e \leq \mathsf{O}(2^{-\sqrt{N}})}$ under successive cancellation decoding. 
\end{theorem}

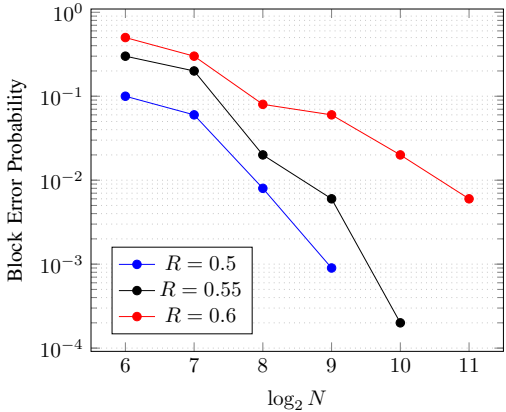
\begin{figure}[!b]
    \centering
    \begin{tikzpicture}[scale = 0.8]
        \begin{semilogyaxis}[
		grid = minor, grid style = dotted,xlabel= $\log_2N$,ylabel=Block Error Probability,  legend entries = {$R = 0.5$, $R = 0.55$, $R = 0.6$}, legend style ={at={(0.05,0.05)},anchor=south west} ]

	\addplot[color=blue,mark=*] coordinates {
		(6,1e-01)
		(7,6e-02)
		(8,8e-03)
		(9,9e-04)
		(10,0)
		(11,0)
	};
 \addplot[color=black,mark=*] coordinates {
		(6,3e-01)
		(7,2e-01)
		(8,2e-02)
		(9,6e-03)
		(10,2e-04)
		(11,0)
	};
  \addplot[color=red,mark=*] coordinates {
		(6,5e-01)
		(7,3e-01)
		(8,8e-02)
		(9,6e-02)
		(10,2e-02)
		(11,6e-03)
	};
	\end{semilogyaxis}%
    \end{tikzpicture}
    \caption{Performance of polar coding on an EQC with SC decoder and an interleaver at $\lambda^* = 0.77$. Note that at $\lambda = 0.77$, $C(\lambda) = 0.54$ in bits/sec which indicates that the capacity is equivalent to $0.7$ bits per channel use.}
    \label{polar_code_results_interleaving}
\end{figure}
\begin{figure}[b!]
    \centering
    \begin{tikzpicture}[scale = 0.8]
        \begin{semilogyaxis}[
		grid = minor, grid style = dotted,xlabel= $\log_2N$,ylabel=Block Error Probability,  legend entries = {$R = 0.5$, $R = 0.55$, $R = 0.6$}, legend style ={at={(0.05,0.05)},anchor=south west} ]

	\addplot[color=blue,mark=*] coordinates {
		(6,2e-01)
		(7,1e-01)
		(8,2e-02)
		(9,4e-03)
		(10,4e-04)
		(11,0)
	};
 \addplot[color=black,mark=*] coordinates {
		(6,5e-01)
		(7,4e-01)
		(8,6e-02)
		(9,2e-02)
		(10,6e-03)
		(11,9e-04)
	};
  \addplot[color=red,mark=*] coordinates {
		(6,7e-01)
		(7,5e-01)
		(8,4e-01)
		(9,1.4e-01)
		(10,9e-02)
		(11,5e-02)
	};
	\end{semilogyaxis}%
    \end{tikzpicture}
    \caption{Performance of conventional polar coding on an EQC without interleaving and SC decoder at $\lambda^* = 0.77$. }
    \label{polar_code_results}
\end{figure}
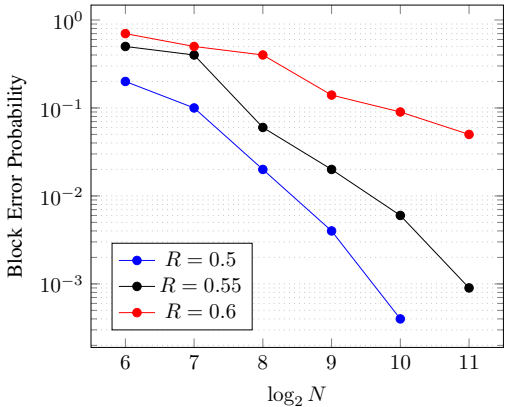

The proof of Theorem~\ref{error_b} relies on the following two theorems. Specifically, for any given arrival rate $\lambda$, let $\mathsf{X}^N$ be the transmitted symbol sequence, $\mathsf{U}^N$ be the polar transformed input sequence, and $\mathsf{Y}^N$ be the received symbol sequence. Further, let $\mathit{I}(;)$ be the general mutual information term, and $\mathit{Z}(\cdot)$ be the Bhattacharyya parameter defined as
$
    Z(C|D) = 2\sum_{d \in \mathcal{D}}\sqrt{P_{C,D}(0,d)P_{C,D}(1,d)},
$
for any two random variables $C \in \{0,1\}$ and $D \in \mathcal{D}$. Then, for an EQC, we have
\begin{theorem}\label{polar}(Polarization)
    At any particular arrival rate $\lambda$ of an erasure queue-channel with erasure probability $p(\cdot)$, for all $0 < \epsilon < 1$,
    \begin{align*}
    \lim_{N \to \infty} \frac{1}{N} |\{i : \mathit{I}(\mathsf{U}_i; \mathsf{U}_1^{i-1},\mathsf{Y}_1^N) < \epsilon\}| &= \mathbb{E}_\pi[p(\mathsf{W})],\\
     \lim_{N \to \infty} \frac{1}{N} |\{i : \mathit{I}(\mathsf{U}_i; \mathsf{U}_1^{i-1},\mathsf{Y}_1^N) > 1 - \epsilon\}| &=  1 - \mathbb{E}_\pi[p(\mathsf{W})], 
\end{align*}
where $i \in \{0,1,\ldots,N-1\}$, and $\pi$ is the stationary distribution of sojourn times of the information bits in the queue.
\end{theorem}

\begin{theorem}\label{fast_polarization}
    (Fast Polarization)  For any arrival rate $\lambda$ of an erasure queue-channel with erasure probability $p(\cdot)$, for all $\beta < 1/2$,
     \begin{align*}
    \lim_{N \to \infty} \frac{1}{N} |\{i : \mathit{Z}(\mathsf{U}_i| \mathsf{U}_1^{i-1},\mathsf{Y}_1^N) < 2^{-N^{\beta}}\}| &= \mathbb{E}_\pi[p(\mathsf{W})].
\end{align*}
\end{theorem}
 Note that the fast polarization of low-entropy set is crucial because Bhattacharyya parameter is known to upper bound the error probability of polar codes \cite[Proposition~2.2]{sasoglu2011polar}.


 \begin{remark}
    The above results assume that the queue is at stationarity --- this can be  achieved by initializing the queue in its stationary distribution, i.e., by starting with some dummy bits.
 \end{remark}
\subsection{Discussion and Challenges}\label{discussion}
The above analytical guarantees are inadequate to cover an EQC with the well-motivated form $p(W)=1-e^{-\kappa W},$ or any continuous, increasing function of $W.$ This is because $W_i$ (and hence $p(W_i)$) turn out to be Markov processes evolving in uncountable state space. Unfortunately, it appears challenging to establish directly the promptly $\psi$-mixing property for such a general EQC. 

One possible workaround could involve approximating the desired $p(\cdot)$ function as a monotone limit of a sequence $\{p_k(\cdot),\ k\ge 1\}$ of `simple functions', i.e., step functions with finitely many steps. If we can prove polarization for the sequence of EQCs governed by erasure functions $\{p_k(\cdot),\ k\ge 1\},$ we can invoke a standard monotone convergence theorem argument and obtain a capacity-achieving sequence of polar codes for the original EQC. Unfortunately, while the $\{p_k(\cdot),\ k\ge 1\}$ evolves in a finite state space, it no longer enjoys Markovity in general! Thus, a proof approach for showing the polarization of a general EQC remains elusive. Our ongoing work explores direct approaches to proving polarization of the EQC by exploiting the underlying renewals in the queue.

\section{Conclusion}\label{conclusion}
We considered an erasure queue-channel (EQC), which has applications in quantum communications, multimedia streaming, and URLLC. Following the capacity results obtained for queue-channels in \cite{mandayam2020classical}, this work aimed at deriving practical channel codes that achieve the capacity for EQC. Our main contributions are as follows: First, we provided a generic wrapper code over a capacity-achieving i.i.d. erasure code to achieve the capacity for an EQC. We derived an essential concentration bound to lower bound the interleaving length in our wrapper technique. Nevertheless, owing to the practical challenges of interleaved systems, we next numerically analyzed the performance of conventional LDPC and Polar codes, used without interleaving over an EQC. Encouraged by the good empirical performance, we proved that conventional Arıkan's polar construction achieves rates arbitrarily close to capacity for a technically restricted class of EQCs. Future directions include proving the theoretical guarantees of Polar and LDPC codes for a general EQC, and designing capacity-achieving codes for queue-channels with other noise models.

\section{Acknowledgements}
JM and KJ acknowledge the Metro Area Quantum Access Network (MAQAN) project, supported by the Ministry of Electronics and Information Technology, India, vide sanction number 13(33)/2020-CC\&BT. This work was also supported, in part, by a grant from Mphasis to the Centre for Quantum Information, Communication, and
Computing (CQuICC). Further, JM gratefully acknowledges support from the Ministry of Education, Government of India, under Prime Minister’s Research Fellowship (PMRF) Scheme, and Nithin Varma Kanumuri for all his helpful discussions.
\bibliography{References} 
\bibliographystyle{ieeetr}

\section{Proofs of Polarization}\label{polarization_proof}
 In the interest of being self contained, we review the general terminology, and a few essential lemmas required to show the polarization of channels with memory under Arıkan's polar transform.
\subsection{Review of polarization for Channels with Memory}\label{memorybackground}
\subsubsection{General Notations}
Let $(\mathsf{X}_i,\mathsf{Y}_i)$, $i \in \mathbb{Z}$ be a stationary and ergodic process, with $\mathsf{X}_i$'s, and $\mathsf{Y}_i$'s indicating the inputs and outputs of a channel $\mathcal{N}$, respectively. We assume that $\mathsf{X}_i$'s are binary and $\mathsf{Y}_i \in \mathcal{Y}$, where $\mathcal{Y}$ is a finite alphabet. Under Arıkan's polar construction, we define the polar transformed inputs as follows: $\mathsf{U}_1^N = \mathsf{X}_1^N F_N G_N,$
where $N=2^n$ is the block length for some $n>0$, $F_N$ is the $N \times N$ bit reversal matrix, and $G_N$ is the $n^{th}$ Kronecker power of the matrix $\begin{psmallmatrix} 1 & 0\\1 & 1\end{psmallmatrix}$. Following \cite{sasoglu}, 
we consider $
    {\mathsf{I}^{\mathbf{b}} = \mathit{I}(\mathsf{U}_i;\mathsf{U}_1^{i-1},\mathsf{Y}_1^N)}
$ and $ {\mathsf{Z}^{\mathbf{b}} = \mathit{Z}(\mathsf{U}_i|\mathsf{U}_1^{i-1},\mathsf{Y}_1^N)},$
where ${\mathbf{b} \in \{0,1\}^{n}}$ is the $(n)$ length binary expansion of $i-1 \in \{0,1,\ldots,N-1\}$. Further, for some i.i.d. Ber(1/2) random variables, we define the random variables ${\mathsf{I}_{n} = \mathsf{I}^{B_1\ldots B_{n}}}$ and ${\mathsf{Z}_{n} = \mathsf{Z}^{B_1\ldots B_{n}}}$ which are uniformly distributed over the sets of $\mathsf{I}^{\mathbf{b}}$'s and $\mathsf{Z}^{\mathbf{b}}$'s. In addition, for simplicity in notation, we define
\begin{align*}
    \mathsf{U}_1^N &= \mathsf{X}_1^N F_N G_N,
    \mathsf{V}_1^N = \mathsf{X}_{N+1}^{2N} F_N G_N,\\
    \mathsf{P}_i &= (\mathsf{U}_1^{i-1},\mathsf{Y}_1^N), 
    \mathsf{Q}_i = (\mathsf{V}_1^{i-1},\mathsf{Y}_{N+1}^{2N}).
\end{align*}

\subsubsection{Polarization Steps}\label{polarization_steps}
The following two lemmas, first established by Arıkan in \cite{arikan2009channel} for memoryless channels and later extended to a class of memory channels by Şaşoğlu in \cite{sasoglu}, are essential in proving the (slow) polarization of channels with memory.
\begin{lemma}\label{con_i}(Convergence of mutual information)
    The sequence $\mathsf{I}_n$ converges almost surely and in $L^1$ to a random variable $\mathsf{I}_{\infty} \in [0,1]$.
\end{lemma}

\begin{lemma}\label{diff_lemma} \cite{sasoglu}[Lemmas~8,11] (Difference across the mutual information terms) 
    If 
    \begin{enumerate}[(i)]
        \item $\mathit{I}(\mathsf{U}_i,\mathsf{V}_i|\mathsf{P}_i,\mathsf{Q}_i) < \epsilon$, and
        \item For all $\xi > 0$, there exists $N_0$ and $\delta(\xi) > 0$ such that for all $N > N_0$ and all $\{0,1\} -$valued random variables $C = f(\mathsf{X}_1^N, \mathsf{Y}_1^N)$ and $D = f(\mathsf{X}_{N+1}^{2N},\mathsf{Y}_{N+1}^{2N})$, ${P_C(1) \in (\xi,1-\xi)}$ implies $P_{C,D}(1,0) > \delta(\xi)$,
    \end{enumerate}
    then $\mathit{I}(\mathsf{U}_i;\mathsf{P}_i) \in (3\xi,1-3\xi)$ implies
     \begin{align*}
        |\mathit{I}(\mathsf{U}_i;\mathsf{P}_i) - \mathit{I}(\mathsf{U}_i+\mathsf{V}_i;\mathsf{P}_i,\mathsf{Q}_i)| > \theta(\xi).  
    \end{align*}
\end{lemma}

It was shown in \cite{sasoglu} that for any stationary and ergodic channels satisfying Lemmas~\ref{con_i} and \ref{diff_lemma}, slow polarization happen under Arıkan's polar construction, i.e., 
$$
 \lim_{N \to \infty} \frac{1}{N} |\{i : \mathit{I}(\mathsf{U}_i; \mathsf{U}_1^{i-1},\mathsf{Y}_1^N) > 1 - \epsilon\}| =  \mathcal{I}(\mathsf{X};\mathsf{Y}),
$$
where $\mathcal{I}(\mathsf{X};\mathsf{Y}) = \lim_{N \to \infty} \frac{1}{N} \mathit{I}(\mathsf{X}_1^N,\mathsf{Y}_1^N)$ holds true. 

Further, it was shown in \cite[Lemma~4.2]{sasoglu2011polar} and \cite{arikan2009rate} that for the fast polarization of low-entropy set to happen, the channel should satisfy the following lemma on Bhattacharyya parameter.

\begin{lemma}\label{z_convergence}
    If $\mathsf{Z}_n$ polarizes to a $\{0,1\}-$random variable $\mathsf{Z}_{\infty}$ and if there exists a $0<k < \infty$ and $\nu_0,\nu_1 >0$ such that for $i = 0,1$,
    $$\mathsf{Z}_{n+1} \leq k \mathsf{Z}_n^{\nu_i} \textrm{ if } B_{n+1} = i, $$
    then
    $$
        \lim_{n \to \infty}\mathbb{P}(\mathsf{Z}_n < 2^{-N^{\beta}}) = \mathbb{P}(\mathsf{Z}_{\infty} = 0),
    $$
    for all $0 < \beta < (\log_2 \nu_0 + \log_2 \nu_1)/2.$
\end{lemma}
In our later sections, we derive the proofs of above lemmas one by one in order to prove Theorem~\ref{error_b} by considering the following realistic assumptions on an EQC.
\subsection{Assumptions}\label{assumptions}
We provide the proof of polarization for an EQC under the following considerations:
\subsubsection{Stationary queue}\label{stationary_assumption}
     We consider that the queue is started at stationary, i.e., we consider $f_{W_0}(w_0) = \pi(w_0),$ where $W_0$ is the sojourn time of the initial information bit in the queue, $f_{W_0}(\cdot)$ is its probability density function, and $\pi(\cdot)$ is the stationary distribution of the sojourn times in the queue. 
     
     We remark that such an assumption is plausible for queue-channels, because transmitting a certain number of dummy information bits into the queue before the actual data transmission can take the queue to stationarity.
\subsubsection{Erasure Probabilities}\label{pwsf} We assume that the sequence of erasure probabilities $\{p(W_i), i \geq 1\}$ evolves according to a finite state Markov process.
\subsection{Proofs of Theorems~\ref{polar} and \ref{fast_polarization}}

    Note that for a stationary queue, Lemma~\ref{con_i} can be easily derived by following the proof steps of \cite[Lemma~7]{sasoglu}; nevertheless, for completeness, we repeat the proof for Lemma~\ref{con_i} in  Appendix Sec.~\ref{step1proof}. Next,

\begin{lemma}\label{psi_mixing_eqc}
    For an erasure queue-channel, under the above assumptions stated in Sec.~\ref{assumptions}, there exists a non-increasing sequence $\psi(N)$, $\psi(N) \to 1$ as $N \to \infty$, such that for any $N \geq N_m \geq N_l \geq 1,$ we have
    $$
        P_{\mathsf{X}_1^{N_l},\mathsf{Y}_1^{N_l},\mathsf{X}_{N_m+1}^{N},\mathsf{Y}_{N_m+1}^{N}} \leq \psi(N_m-N_l) P_{\mathsf{X}_1^{N_l},\mathsf{Y}_1^{N_l}} P_{\mathsf{X}_{N_m+1}^{N},\mathsf{Y}_{N_m+1}^{N}},
    $$
    and $\psi(0) < \infty.$
\end{lemma}
\begin{proof}
    The proof of this lemma directly follows from the proof of Lemma~5 in \cite{shuval2018fast}. Refer to Appendix~C in \cite{shuval2018fast} for the detailed proof.
\end{proof}

Now, following Lemma~\ref{psi_mixing_eqc}, we see that Lemmas~\ref{con_i} Lemma~\ref{diff_lemma} can be easily proven for an EQC following the similar steps as in the proofs of Lemmas~8 and~10 from \cite{sasoglu}. Further, Şaşoğlu and Tal in \cite{sasoglu} proved that Lemmas~8 and~10 are sufficient to prove the slow polarization of channels with memory, which completes the proof of Theorem~\ref{polar}.

Next, for Theorem~\ref{fast_polarization}, following Lemma~\ref{psi_mixing_eqc} and proof steps of \cite[Theorem~2]{sasoglu}, we can show Lemma~\ref{z_convergence} for an EQC. Finally, as in the slow polarization result, Şaşoğlu and Tal proved that Lemma~\ref{z_convergence} is sufficient to show the fast polarization of channels with memory, which completes the proof of Theorem~\ref{fast_polarization} under the assumptions stated in Sec.\ref{assumptions}.

We remark that the step-$(ii)$ in Lemma~\ref{diff_lemma}, can be proven for any general EQC. We use the renewal structures of the queue to prove the result. See Appendix Sec.\ref{alternative_proof}. However, proving step-$(i)$ in Lemma~\ref{diff_lemma} and Lemma~\ref{z_convergence} requires an additional set of tools like $\psi-$mixing of the queues as stated in Sec.~\ref{discussion}. A more detailed study of queue-channels is required to prove the polarization of an EQC in general. Nevertheless, the simulations are motivating enough to state that an EQC does polarize in general.
\section{Appendix}\label{appendix}
\subsection{The Queuing bound:}\label{gen_queue_bound}

\begin{lemma}
For any $G/G/1$ queue, if $\Theta_i$'s are i.i.d. sub-exponential random variables with parameters $(\nu,b)$ and mean $\upsilon < \infty$, then the total number of jobs $J(\dot{T}_b)$ in a `typical' busy period $T_b$ goes down exponentially fast. Specifically, we have
\begin{align*}
    P(J(\dot{T}_b) \geq \eta) &\leq e^{\frac{-\eta \upsilon^2}{2\nu^2} }, \text{ for } 0 < \upsilon < \nu^2/b,\\
                     & \leq e^{\frac{-\eta \upsilon}{2b}}, \text{ for } \upsilon > \nu^2/b.
\end{align*}
\end{lemma}

\subsection{Proof of Lemma~\ref{concentration}}\label{concentrationlemmaproof}

Recall that in an $M/M/1$ queue, the inter-arrival times $A_i$'s are i.i.d. exponential rv's with mean $\mathbb{E}[A] = 1/\lambda$, and $S_i$'s are i.i.d. exponential rvs with mean $\mathbb{E}[S] = 1/\mu$. Define a rv sequence $\{\Theta_i, i \geq 0\}$ such that $\Theta_i = A_i - S_i$, and $\mathbb{E}[\theta] = \mathbb{E}[A] - \mathbb{E}[S]$. Now, consider the probability that the number of jobs in a typical busy period $\dot{T}_b$ denoted by $J(\dot{T}_b)$ is larger than $l$. Following the queuing dynamics, we have
\begin{align*}
        P\{J({\dot{T}_b}) \geq l\} &= P(\textstyle \sum_{i=1}^k \Theta_i < 0, 1 \leq k \leq l)\\
    &\leq P(\textstyle \sum_{i=1}^{l} (\Theta_i - \mathbb{E}[\Theta_i]) < -l \mathbb{E}[\Theta]) \\
    &\leq \inf_{t \geq 0}\frac{\mathbb{E}[e^{-t \textstyle \sum_{i=1}^{l} (\Theta_i - \mathbb{E}[\Theta_i])]}}{e^{t (l) \mathbb{E}[\Theta]}}\\
    &= \inf_{t \geq 0} \mathbb{E}[\textstyle \prod_{i=1}^{l} \exp({-t \Theta_i})] 
    = \left(\frac{\lambda \mu}{(\frac{\lambda + \mu}{2})^2}\right)^{l},
\end{align*}
where the last inequality follows from Chernoff bound completing the proof. \qed

\subsection{Proof of Lemma~\ref{con_i}:}\label{step1proof}
Consider a probability space $(\Omega,\mathscr{F},\mathbb{P})$, with $\Omega$ being the set of all binary sequences $(b_1,b_2,\ldots) \in \{0,1\}^{\infty}$, $\mathscr{F}$ being the Borel field generated by the cylinder sets defined as $\Sigma(b_1,\ldots,b_l) = \{\omega \in \Omega : \omega_1 = b_1, \ldots, \omega_l = b_l\}$, $\forall l \geq 1,$ such that $b_1, \ldots, b_l \in \{0,1\},$ and $\mathbb{P}$ is the probability measure defined on $\mathscr{F}$ such that $\mathbb{P}(\Sigma(b_1,\ldots,b_l) = 1/2^l.$ For each $n \geq 1$, we define $\mathscr{F}_n$ as the Borel field generated by the cylinder sets $\Sigma(b_1,\ldots,b_i), 1 \leq i \leq l$, and $\mathscr{F}_{0}$ is defined as the trivial Borel field containing the nullset and $\Omega$. It can be easily seen by the construction of above Borel fields that, $\mathscr{F}_0 \subset \mathscr{F}_1 \subset \ldots \subset \mathscr{F}.$ The above random processes can now be formally defined as follows: For any $\omega = \{\omega_1,\omega_2,\ldots\} \in \Omega$, we define $\mathbf{b}(\omega) = \omega_1,\omega_2,\ldots,\omega_l$, and $I_n(\omega) = I_n^{\mathbf{b}(\omega)}$, $I_0 = I(\mathcal{N})$.

Now, the proof of Lemma~\ref{con_i} proceeds by proving that the sequence $\{I_n,n\geq 0\}$ is a bounded sub martingale i.e., we first show that the above martingale construction satisfies the following:
\begin{enumerate}[(i)]
    \item $\mathscr{F}_n \subset \mathscr{F}_{n+1},$ and $I_n$ is $\mathscr{F}_n$ measurable.
    \item $\mathbb{E}[|I_n|] < \infty$.
    \item $I_n \leq \mathbb{E}[I_{n+1}|I_n]$
\end{enumerate}

We see that (i) follows directly from the construction of the Borel fields and the definition of $I_n$'s, (ii) is true from the fact that $0 \leq I_n \leq 1$. 

Note that from the polarization construction we have
\begin{align*}
    I_{n+1} &= I(U_i+V_i;P_i,Q_i) , &\text{ if } B_{l+1} = 0, \\
            &= I(U_i;P_i, Q_i, U_i + V_i), &\text{ if } B_{l+1} = 1.
\end{align*}

Now, consider $\mathbb{E}[I_{n+1}|\Sigma(b_1,\ldots,b_l)]$. From the polarization construction, we have
\begin{align}
\begin{split}
     &\mathbb{E}[I_{n+1}|\Sigma(b_1,\ldots,b_l)] \\ &= \frac{1}{2}[I(U_i+V_i;P_i,Q_i) + I(V_i; P_i,Q_i,U_i+V_i)] \\
&= \frac{1}{2}[H(U_i+V_i) - H(U_i+V_i|P_i,Q_i) \\
& \hspace{0.5in}+ H(V_i) - H(V_i|P_i,Q_i,U_i+V_i) \\
&\geq \frac{1}{2}[H(U_i)+H(V_i) - H(U_i,V_i|P_i,Q_i)] \\
&\geq \frac{1}{2}[H(U_i)+H(V_i) - H(U_i|P_i) - H(V_i|Q_i)] \\
&= \frac{1}{2}[I(U_i;P_i) + I(V_i;Q_i)] \\
&= I(U_i;P_i) = I_n\\
\end{split}
\end{align}

Thus, by the general convergence results of the martingales \cite[Theorem~9.4.5]{chung2001course}, the sequence $I_0, I_1, \ldots$ converges almost surely and in $L^1$ to a random variable $I_{\infty}.$ \qed

\subsection{Alternative proof of step$-(ii)$ in Lemma~\ref{diff_lemma}:}\label{alternative_proof}
   The following lemma is essential in proving step$-(ii)$ of Lemma~\ref{diff_lemma}. Let $\{E_i, i > 0\}$ be the erasure sequence corresponding to the output sequence $\{Y_i, i > 0\}$ defined as $E_i = 1$ if $Y_i$ is an erasure and $E_i = 0$ otherwise. Then following Lemma~\ref{concentration}, we have 

   \begin{lemma}\label{independent}
       For any $\alpha > 0$ if $N = 2^n > \frac{\ln{1/\alpha}}{2\ln{(\lambda+\mu)} - \ln{(4\lambda \mu)}}$, the there exists a $\delta(\alpha)$ such that the erasure sequences $E_1^N$ and $E_{2N+1}^{3N}$ are $\delta-$independent, i.e., they follow
       \begin{align}\label{deltainde}
           |P_{E_1^N,E_{2N+1}^{3N}} - P_{E_1^N} P_{E_{2N+1}^{3N}}| < \delta.
       \end{align}
   \end{lemma}
   
      \begin{proof}
    In order the prove the $\delta-$independence in equation~\eqref{deltainde}, first note that the sequences $E_1^N$ and $E_{2N+1}^{3N}$ are independent if and only if there exists a renewal between the information bits $N+1$ and $2N$. 
    
    Next, let $R$ be the event defining that there exists a renewal between $N+1$ and $2N$; further, define $T_{N+1}$, as the busy period duration in which $(N+1)^{th}$ information bit has occurred, $L(T_{N+1})$, as the residual number of information bits after $(N+1)^{th}$ information bit in the renewal period $T_{N+1}$,   and $J({T_{N+1}})$ as the total number of information bits in the renewal period $T_{N+1}$. Recall, $N = 2^n > \frac{\ln{1/\alpha}}{2\ln{(\lambda+\mu)} - \ln{(4\lambda \mu)}}$, for some $n>0$ and $\alpha > 0$.  Then, following the above notation, we have

    \begin{align}\label{norenewalprob}
        \begin{split}
            P(R^c) &\numeq{a} P\{N \leq  L({T_{N+1}})\} 
            \numleq{b} P\{N \leq J({T_{N+1}})\} \\
            &\numleq{c} \sum_{j = N+1}^{\infty} j \frac{P\{J({\dot{T}_{b}) = j\}}}{\mathbb{E}[J(\dot{T}_b]}
            \numleq{d} \delta(\alpha).
        \end{split}
    \end{align}

We now provide the justification for each of the above inequalities.
\begin{itemize}
    \item First, $(a)$ holds because there exists no renewal between $N+1$ and $2N+1$ if and only if the information bit $2N+1$ falls below the residual duration of that particular renewal after the $N+1$ information bit has occurred. In other words, the $N^{th}$ information bit after $(N+1)^{th}$ information bit (which is $(2N+1)^{th}$ bit in the original sequence) must be less than the residual number of information bits after $N+1$ in the busy period $T_{N+1}$.
    \item Next, $(b)$ holds 
    because the event $N \leq L(T_{N+1})$ is contained in the event $N \leq J(T_{N+1}).$
    \item The inequality $(c)$ occurs because the busy period in which the $(N+1)^{th}$ information bit has occurred could be atypically large \cite{gallager2013}. By standard sampling arguments, the probability distributions of the number of information bits in a busy period of tagged information bit $k$ denoted by $J({T_{k}})$ is related to a number of information bits in any general typical busy period denoted by $J({\dot{T}_b})$ as follows:
    $$
        P(J({T_k}) = i) = \frac{iP(J({\dot{T}_b})=i)}{\mathbb{E}[J({\dot{T}_b})]}.
    $$
    \item Finally, $(d)$ is true from the choice of N and Lemma~\ref{concentration}.
\end{itemize}

Now, considering the LHS of equation~\eqref{deltainde}, we have
\begin{align*}
    & |P_{E_1^N,E_{2N+1}^{3N}} - P_{E_1^N} P_{E_{2N+1}^{3N}}| \\
    & = |P_{E_1^N,E_{2N+1}^{3N}|R}.P(R) + P_{E_1^N,E_{2N+1}^{3N}|R^c}P(R^c) \\
    &\hspace{0.5in} - P_{E_1^N|R} P_{E_{2N+1}^{3N}|R} . (P(R))^2 \\
    & \hspace{0.5in} - P_{E_1^N|R^c} P_{E_{2N+1}^{3N}|R^c} . (P(R^c))^2|\\
    &\numeq{a} |P_{E_1^N} P_{E_{2N+1}^{3N}} (P(R) - P^2(R)) | \\
    &\hspace{0.5in} + |P_{E_1^N,E_{2N+1}^{3N}|R^c}P(R^c)| \\
    &\hspace{0.5in}  + |P_{E_1^N|R^c} P_{E_{2N+1}^{3N}|R^c} P^2(R^c)| \\
    &\numleq{b} 3P(R^c) \numleq{c} \delta(\alpha),
\end{align*}
where $(a)$ holds true because, given $R$, the events $E_1^N$ and $E_{2N+1}^{3N}$ are independent and noting that $|a-b| \leq |a| + |b|$ . $(b)$ is true by upper bounding the probabilities by 1, and finally, $(c)$ follows from equation~\eqref{norenewalprob} completing the proof.                                          
\end{proof}
  
Now, using the above two lemmas, we prove the step $(ii)$ of Lemma~\ref{diff_lemma}. Define $C = f(X_{2N+1}^{3N},Y_{2N+1}^{3N})$ as another $\{0,1\}$ random variable. Then, we have
\begin{align*}
  2p_{A,B}(1,0) &= p_{A,B}(1,0) + p_{B,C} (1,0)  \geq p_{A,C}(1,0)\\               
                &= p_A(1) - p_{A,C}(1,1) \\
                & \numgeq{a} p_A(1)(1-p_C(1)) - \delta \\
                & = p_A(1)(1-p_A(1)) - \delta \\
                &\geq \xi(1-(1-\xi)) - \delta \\
                &= \xi^2 - \delta,
\end{align*}
where $(a)$ follows from Lemma~\ref{independent}. Further, for any $\delta < \xi^3$, we see that the above probability $p_{A,B}(1,0)$ being lower bounded by ${\theta(\xi) = \xi^2(1 - \xi)/2}$. This, in turn, implies that there exists an $N_0$ such that Lemma~\ref{diff_lemma} holds true.\qed

\subsection{Performance of $(2,1)$ Repetition Codes for an EQC}\label{two-onerep}
In a $(2,1)$ repetition code, we assume each classical information bit $i$ is repeated twice and encoded into two consecutive information bits. We say that the $i^{th}$ information bit is erased iff both $i$ and the 
  $(i+1)^{th}$ information bits are erased. Let $p_{e}$ denote the probability of error in a repetition code. For simplicity of analysis, we provide the error probability in case of an $M/M/1$ (a single server queue determined by a Poisson arrival process and i.i.d. exponential service times); however, similar steps hold true for any general queue in the system.

 Note that since we are analyzing $M/M/1$ queue, $A_i$'s are i.i.d. exponential distributed rv's with rate $\lambda$ and $S_i$'s are i.i.d. exponential rv's with rate $\mu$. Then, 
\begin{lemma}\label{lemma_jd}
The conditional probability distribution of any two consecutive waiting times $W_i, W_{i+1}, $ in an M/M/1 queue is given by
$$
     f_{W_{i+1}|W_i}(w_{i+1}|w_i) = \frac{\mu}{\lambda+\mu} [\mu e^{-\lambda w_i - \mu w_{i+1}} + \lambda g(w_i,w_{i+1})],
$$
where $g(w_i,w_{i+1})$ is defined as
\begin{align}
  \begin{split}
        g(w_i,w_{i+1}) &= \begin{cases}
                e^{-\mu(w_{i+1}-w_i)}, & w_i < w_{i+1},\\
                e^{\lambda(w_{i+1}-w_i)}, & w_i > w_{i+1}.
        \end{cases}
  \end{split}
\end{align}
\end{lemma}
\begin{proof} 
Note that any two consecutive waiting times $W_i$, $W_j$, in an M/M/1 queue are related to each other by Lindley's recursion. Hence, we have
$$
W_{j} = \max(W_i - A_{j}, 0) + S_{j},
$$
where $A_j$ and $S_j$ 
Consequently, the conditional CDF of $W_j$ given $W_i$ can be written as
\begin{align}\label{cdf}
\begin{split}
     F_{W_j|W_i}(w_j|w_i) &= P(\{S_j - A_j \leq w_j - w_i\} \cap \{A_j < w_i\}) \\
    & \hspace{0.3in}+ P(\{A_j > w_i\} \cap \{S_j \leq w_j\}).
\end{split}   
\end{align}      
We now evaluate the first term in the above expression case by case. Considering the first case, i.e., when $w_i < w_j$, the first term in the above expression can be evaluated to be
\begin{align*}
    & P(\{S_j - A_j \leq w_j - w_i\} \cap \{A_j < w_i\})\\ &= \int_{0}^{w_1} P(S_j \leq w_j-w_i+x) f_A(x) dx\\
    &= \int_{0}^{w_i} (1 - e^{-\mu (w_j-w_i+x)}) \lambda e^{-\lambda x} dx\\
    &= (1 - e^{-\lambda w_i}) - \frac{\lambda}{\lambda+\mu}[e^{-\mu(w_j-w_i)} - e^{-\mu w_j - \lambda w_i}].
\end{align*}
Next, in the other case i.e., when $w_i > w_j$, we have
\begin{align*}
     & P(\{S_j - A_j \leq w_j - w_i\} \cap \{A_j < w_i\})\\
     &= \int_{w_i-w_j}^{w_i} P(S_j \leq w_j-w_i+x) f_A(x) dx\\
     &= \int_{w_i-w_j}^{w_i} (1 - e^{-\mu (w_j-w_i+x)}) \lambda e^{-\lambda x} dx\\
     &= \frac{\mu}{\lambda+\mu}e^{-\lambda(w_i-w_j)} - e^{-\lambda w_i} + \frac{\lambda}{\lambda+\mu} e^{-\mu w_j-\lambda w_i}.
\end{align*}

Further, using the independence of $S_j$ and $A_j$, the second term in equation~\eqref{cdf} can be easily evaluated to be ${e^{-\lambda w_i}(1- e^{-\mu w_j})}$.
Substituting, the above expressions for both the cases in equation (2), and evaluating we get the conditional CDF as follows:
\begin{align*}
    \begin{split}
       &= (1 - e^{-\lambda w_i}) - \frac{\lambda}{\lambda+\mu}[e^{-\mu(w_j-w_i)} - e^{-\mu w_j - \lambda w_i}] \\
       &\hspace{0.5in} + {e^{-\lambda w_i}(1- e^{-\mu w_j})}, \text{ when  } w_i < w_j, \\
       &= \frac{\mu}{\lambda+\mu}e^{-\lambda(w_i-w_j)} - e^{-\lambda w_i} + \frac{\lambda}{\lambda+\mu} e^{-\mu w_j-\lambda w_i} \\
       &\hspace{0.5in} + {e^{-\lambda w_i}(1- e^{-\mu w_j})}, \text{ when  }  w_i > w_j.\\
    \end{split}
\end{align*}
Finally, differentiating the above conditional CDF's w.r.t. $w_j$ proves the result.
\end{proof}

\begin{theorem}\label{thm_de2}
When the coherence times are exponential, i.e., when $p(W)$ takes the form $1 - e^{-\kappa W}$, the probability of decoding error, $p_e$, of a (2,1) repetition code over an erasure queue-channel in $M/M/1$ queue is given by 
\begin{align*}p_e &= 
    \frac{\kappa^2[2\kappa^2 - \lambda^2 + 2\lambda \mu + \mu^2 + \kappa (\lambda+3\mu)]}{(2\kappa+\mu-\lambda)(\kappa+\mu-\lambda)(\kappa+\mu)^2}.
\end{align*}
\end{theorem}
\begin{proof}
    The decoding error of (2,1) repetition code in the M/M/1 queue-channel can be seen as the probability of erasure of any two consecutive symbols in the queue. The probability of erasure of two consecutive erasures, $p(E_{i} \cap E_{j})$ in an erasure queue-channel is given by
\[
\begin{split}
    p(E_i \cap E_{j}) &= \int_{0}^{\infty} \int_{0}^{\infty} p(E_{i} \cap E_{j}|W_i, W_{j}) \\
    & \hspace{0.8in} f_{W_i, W_{j}}(w_i, w_{j}) dw_i dw_{j} \\
    &= \int_{0}^{\infty} \int_{0}^{\infty} (1 - e^{-\kappa w_i})(1 - e^{-\kappa w_{j}}) \\
    & \hspace{0.8in} f_{W_i}(w_i) f_{W_{j}|W_i}(w_j|w_i) dw_i dw_{j}\\
    &=\int_{0}^{\infty} (1 - e^{-\kappa w_i}) f_{W_i}(w_i) \{\int_{0}^{\infty} (1 - e^{-\kappa w_j}) \\
    & \hspace{1in} f_{W_{j}|W_i}(w_j|w_i) dw_j\} dw_{i}\\
    &\numeq{a} \int_0^{\infty} (1 - e^{-\kappa w_i}) (\mu-\lambda) e^{-(\mu-\lambda)w_i} \\
     & \hspace{0.6in} \left[ 1 - \frac{(\mu \kappa e^{-\lambda w_i} - \lambda \mu e^{-\kappa w_i})}{(\kappa+\mu)(\kappa - \lambda)}\right] dw_i,\\
\end{split}
\]
where $(a)$ follows by substituting the conditional distribution from Lemma~\ref{lemma_jd}, and using the fact that marginal distribution of waiting times in an $M/M/1$ queue is an exponential distribution with rate $\mu-\lambda$. Further, solving the integration in $(a)$ w.r.t. $w_i$ provides us the desired result. 
\end{proof}

\end{document}